\documentclass[12pt]{article}    
\usepackage[ansinew]{inputenc} 
\usepackage{hyperref}
\usepackage{graphics}
\usepackage{epsfig}
\usepackage{epstopdf}
\usepackage{graphicx} 
\usepackage{amssymb}  
\usepackage{amsmath}
\usepackage{multirow} 
\usepackage{enumitem} 
\usepackage{color}     
\usepackage{indentfirst}  
\usepackage{booktabs}
\usepackage{caption} 
\usepackage{appendix}
\renewcommand\appendix{\section*{Appendix}}
\usepackage[gen]{eurosym}
\usepackage[scientific-notation=true]{siunitx}
\begin{document}
\newtheorem{definition}{Definition}
\newtheorem{theorem}{Theorem}
\newtheorem{example}{Example}
\newtheorem{corollary}{Corollary}
\newtheorem{lemma}{Lemma}
\newtheorem{proposition}{Proposition}
\newenvironment{proof}{{\bf Proof:\ \ }}{\qed}
\newcommand{\qed}{\rule{0.5em}{1.5ex}}
\newcommand{\bfg}[1]{\mbox{\boldmath $#1$\unboldmath}}

\begin{center}

\section*{A Generalization of the Power Law Distribution with Nonlinear Exponent}

\vskip 0.2in {\sc \bf Faustino Prieto\footnote{Corresponding author. Tel.: +34 942 206758; fax: +34 942 201603. E-mail address: faustino.prieto@unican.es (F. Prieto).}, Jos\'e Mar\'{\i}a Sarabia
\vskip 0.2in

{\small\it Department of Economics, University of Cantabria, Avenida de los Castros s/n, 39005 Santander, Spain.
}\\
}

\end{center}

\begin{abstract}\noindent
The power law distribution is usually used to fit data in the upper tail of the distribution. However, commonly it is not valid to model data in all the range. In this paper, we present a new family of distributions, the 
so-called Generalized Power Law (GPL),  which can be useful for modeling data in all the range and possess power law tails. To do that, we model the exponent of the power law using a non-linear function which depends on data and two parameters. Then, we provide some basic properties and some specific models of that new family of distributions. After that, we study a relevant model of the family, with special emphasis on the quantile and hazard functions, and the corresponding estimation and testing methods. Finally, as an empirical evidence, we study how the debt is distributed across municipalities in Spain. We check that power law model is only valid in the upper tail; we show analytically and graphically the competence of the new model with municipal debt data in the whole range; and we compare the new distribution with other well-known distributions including the Lognormal, the Generalized Pareto, the Fisk, the Burr type XII and the Dagum models.
\end{abstract}
\vskip 0.2in

\noindent {\bf Key Words}: Power law behavior; Whole range fitting; Complexity; Municipal debt

\section{Introduction}

Many empirical analysis of diverse real phenomena (the population of the cities, the annual income of the people,
the solar flare intensity, the failures in power grids, the protein interaction degree, etc) have confirmed the power
law behavior in the upper tail of their distributions - the largest values of the variable of interest, above a certain lower bound, can be modeled
statistically by a classical Pareto distribution, with shape parameter $\alpha$ also known as exponent of the power law or simply constant $\alpha$ (see, for example, \cite{Pinto,Clauset,Newman2005,Clementi2006,Rosas}).
That empirical evidence comes with many advantages: it can help us to understand the
underlying data generating process \cite{Mayo}; it gives us tools for computer simulation of those phenomena \cite{Kelton2000}; etc.

However, Pareto distribution is not usually valid to model those real phenomena in the whole range -
if we consider high, medium and low ranges of those data all together, the power law behaviour usually disappears.
For example: failures in power grids can be described by the Lomax distribution \cite{Prieto2014,Cuadra2015};
or in the case of protein interaction networks of three species (C.elegans, S.cerevisiae and E.coli), or in the case of the metabolic networks with human and yeast data,
the lognormal distribution provides the best description for the empirical data \cite{{Stumpf2005}}, in the whole range.

The Pareto distribution hierarchy, composed by Pareto type I (Power Law), Pareto type II (with Lomax distribution as a special case),
Pareto type III and type IV, is a well known extension of the power law \cite{Arnold1983,Arnold2015}.  Those family of distributions, also known as Generalized Pareto distributions, have
extended the scope of the classical Pareto model, as for example, with the failures in power grids and the Lomax distribution, as mentioned previously.

The aim of this study is twofold. Firstly, to explore the properties of a new family of GPL distributions that we could use to model real phenomena in the whole range, phenomena with power law tail. Secondly, to provide empirical evidence of the efficacy of those distributions with real datasets.
Our primary hypothesis was that Pareto shape parameter, the exponent $\alpha$, is not constant and varies according to a non-linear function $g$ which depends on data
\cite{Sarabia2009}. We found a surprisingly rich family of distributions, with only three parameters, which includes Pareto and Pareto Positive Stable (PPS) distributions as special cases,
and we also found that a new distribution, a relevant model of that family, is a good alternative for modeling debt data of the indebted municipalities in Spain in the whole range.

The rest of this paper is organized as follows: in Section 2, we introduce a new family of GPL distributions;
in Section 3, we present a new distribution, which belongs to that new family;
an empirical application of that new distribution to municipal debt with Spanish data is included in Section 4;
finally, the conclusions are given in Section 5.

\section{A new family of Generalized Power Law distributions}\label{secfamily}

In this section we obtain the new family of Generalized Power Law (GPL) distributions. Our idea is to construct an extension of the Power law, where the exponent is not constant and is modeled by using a non-linear function of the data. Then, let consider a real function $g:(1,\infty)\rightarrow\mathbb{R}^+$ continuous, positive and differentiable on $(1,\infty)$ satisfying the following conditions,
\begin{equation}\label{c1}
\displaystyle\lim_{z \to 1^+}z^{g(z)}=1\;\mbox{and}\;\displaystyle\lim_{z \to\infty}z^{g(z)}=\infty,
\end{equation}
and
\begin{equation}\label{c2}
\displaystyle\frac{g'(z)}{g(z)}>\displaystyle\frac{-1}{z\log(z)},\;\forall\;z>1.
\end{equation}
Now, using $g(\cdot)$, we define the function
\begin{equation}\label{cdf}
F(x)=1-\left(\frac{x}{\sigma}\right)^{-g(x/\sigma)},\hspace{0.1cm}x>\sigma,
\end{equation}
and $F(x)=0$ if $x\leq\sigma$. Note that $\sigma$ is a scale parameter. We have the following Theorem.
\begin{theorem}
Let consider the functional form defined in (\ref{cdf}), where the function $g(\cdot)$ satisfies conditions (\ref{c1}) and (\ref{c2}). Then (\ref{cdf}) is a genuine cumulative distribution function (cdf).
\end{theorem}
\begin{proof}
It is direct to check that $F(-\infty)=0$, $F(\infty)=1$ and $F(x)$ is right continuous. Finally, $F(x)$ is nondecreasing since,
$$g\left(\frac{x}{\sigma}\right)+\frac{x}{\sigma}\log\left(\frac{x}{\sigma}\right)\;g'\left(\frac{x}{\sigma}\right)>0,\;\forall x>\sigma,$$
using condition (\ref{c2}).
\end{proof}

The family of distributions define in Eq.(\ref{cdf}) includes the classical Pareto distribution (also known as Power Law or Pareto type I distribution)  \cite{Arnold1983,Pareto1964} as a special case when $g(z)=\alpha,\; \forall \;z> 1$ with $\alpha>0$.

\subsection{Basic Properties}\label{basics}

The survival function $S(x)=\Pr(X>x)=1-F(x)$ is given by:
$$S(x)=\left(\frac{x}{\sigma}\right)^{-g(x/\sigma)},\;x>\sigma,$$
and $S(x)=1$ if $x\leq\sigma$.

The probability density function (pdf) of that family of distributions is given by,
$$f(x)=\displaystyle\frac{dF(x)}{dx}=\left\{g\left(\frac{x}{\sigma}\right)+\left(\frac{x}{\sigma}\right)\log\left(\frac{x}{\sigma}\right)\;g'\left(\frac{x}{\sigma}\right)\right\}\frac{S(x)}{x},$$
if $x>\sigma$ and $f(x)=0$ if $x\leq\sigma$.

The hazard function $h(x)=\frac{f(x)}{\Pr(X>x)}=\frac{f(x)}{S(x)}$ is as follows:
$$h(x)=\displaystyle\frac{g\left(\displaystyle\frac{x}{\sigma}\right)+\left(\displaystyle\frac{x}{\sigma}\right)\log\left(\displaystyle\frac{x}{\sigma}\right)g'\left(\displaystyle\frac{x}{\sigma}\right)}{x},$$
if $x>\sigma$ and $h(x)=0$ if $x\leq\sigma$. Some graphics of this family  are included in Section \ref{relevant} for a relevant special case.

\subsection{Some models of generalized power law and extensions}

In this section we present some specific models of Generalized Power Law distributions and we also provide some extensions of that new family of distributions. To model the $g(\cdot)$ function, we choose some flexible functions which depend on two parameters $\alpha$ and $\beta$, and that include as special case the constant function by setting $\beta=0$.

Table \ref{distexamples} provides some models of Generalized Power Law distributions, where we have reported the $g(z)$ function, the survival function and the pdf. The simplest choice, that is, $g(z)=\alpha$, corresponds to the usual power law, or classical Pareto distribution. The choice $g(z)=\alpha\log^\beta(z)$ corresponds to the PPS distribution \cite{Sarabia2009,Guillen2011}. As far as we know, the rest of models are new.  

\begin{table}[p]\scriptsize
\renewcommand{\tablename}{\footnotesize{Table}}
\caption{\label{distexamples}\footnotesize Some examples of distributions which belongs to the new family of distributions described.}
\centering
\begin{tabular}{@{}l c l l @{}}
\toprule
$g(z)$&  $\beta$&$ S(x)$& $f(x)$\\
\midrule
$\alpha$&&
$\left(\displaystyle\frac{x}{\sigma}\right)^{-\alpha}$&
$\alpha\displaystyle\frac{S(x)}{x}$\\[3mm]

$\alpha\log^\beta(z)$&
$\beta>-1$&
$\left(\displaystyle\frac{x}{\sigma}\right)^{-\alpha\log^\beta(x/\sigma)}$&
$\alpha(\beta+1)\log^\beta(x/\sigma)\displaystyle\frac{S(x)}{x}$\\[3mm]

$\alpha z^\beta$&
$\beta\geq0$&
$\left(\displaystyle\frac{x}{\sigma}\right)^{-\alpha (x/\sigma)^\beta}$&
$\alpha [1+\beta\log(x/\sigma)]\left(\displaystyle\frac{x}{\sigma}\right)^\beta\displaystyle\frac{S(x)}{x}$\\[3mm]

$\alpha+\beta \log z$&
$\beta\geq0$&
$\left(\displaystyle\frac{x}{\sigma}\right)^{-\alpha-\beta \log(x/\sigma)}$&
$[\alpha+2\beta\log(x/\sigma)]\displaystyle\frac{S(x)}{x}$\\[3mm]

$\alpha+\beta z$&
$\beta\geq0$&
$\left(\displaystyle\frac{x}{\sigma}\right)^{-\alpha-\beta (x/\sigma)}$&
$[\alpha+\beta(x/\sigma)(1+\log(x/\sigma))]\displaystyle\frac{S(x)}{x}$\\[3mm]

$\alpha-\beta\left(\displaystyle\frac{z-1}{z\log z}\right)$&
$\beta\leq\alpha$&
$\left(\displaystyle\frac{x}{\sigma}\right)^{-\alpha+\beta\left[\frac{(x/\sigma)-1}{(x/\sigma)\log(x/\sigma)}\right]}$&
$\left[\alpha-\displaystyle\frac{\beta}{(x/\sigma)}\right]\displaystyle\frac{S(x)}{x}$\\[5mm]

$\alpha-\displaystyle\frac{\beta}{z}$&
$\beta\leq\alpha$&
$\left(\displaystyle\frac{x}{\sigma}\right)^{-\alpha+\beta\sigma/x}$&
$\left[\alpha+\displaystyle\frac{\beta}{(x/\sigma)}(\log(x/\sigma)-1)\right]\displaystyle\frac{S(x)}{x}$\\[5mm]

$\alpha+\beta\left(\displaystyle\frac{z-1}{\log(z)}\right)$&
$\beta\geq0$&
$\left(\displaystyle\frac{x}{\sigma}\right)^{-\alpha-\beta\left[\frac{(x/\sigma)-1}{\log(x/\sigma)}\right]}$&
$\left[\alpha+\beta(x/\sigma)\right]\displaystyle\frac{S(x)}{x}$\\[5mm]

$\alpha+\beta\left(\displaystyle\frac{\log z}{1+\log z}\right)$&
$\beta\geq-\alpha$&
$\left(\displaystyle\frac{x}{\sigma}\right)^{-\alpha-\beta\left[\frac{\log(x/\sigma)}{\log(x/\sigma)+1}\right]}$&
$\left[\alpha+\beta\displaystyle\frac{\log(x/\sigma)[\log(x/\sigma)+2]}{[\log(x/\sigma)+1]^2}\right]\displaystyle\frac{S(x)}{x}$\\[5mm]

$\alpha+\beta\left(\displaystyle\frac{z-1}{z}\right)$&
$\beta\geq-\alpha$&
$\left(\displaystyle\frac{x}{\sigma}\right)^{-\alpha-\beta\left[\frac{(x/\sigma)-1}{(x/\sigma)}\right]}$&
$\left[\alpha+\beta\displaystyle\frac{\log(x/\sigma)-1+(x/\sigma)}{(x/\sigma)}\right]\displaystyle\frac{S(x)}{x}$\\[5mm]

$\alpha\left(\displaystyle\frac{\log z}{1+\log z}\right)^\beta$&
$\beta>-1$&
$\left(\displaystyle\frac{x}{\sigma}\right)^{-\alpha\left[\frac{\log(x/\sigma)}{\log(x/\sigma)+1}\right]^{\beta}}$&
$\alpha\left[\displaystyle\frac{\log(x/\sigma)+1+\beta}{\log(x/\sigma)+1}\right]\left[\displaystyle\frac{\log(x/\sigma)}{\log(x/\sigma)+1}\right]^\beta\displaystyle\frac{S(x)}{x}$\\[5mm]

$\alpha\left(\displaystyle\frac{z-1}{z}\right)^\beta$&
$\beta\geq-1$&
$\left(\displaystyle\frac{x}{\sigma}\right)^{-\alpha\left[\frac{(x/\sigma)-1}{(x/\sigma)}\right]^{\beta}}$&
$\alpha\left[\displaystyle\frac{(x/\sigma)-1+\beta\log(x/\sigma)}{(x/\sigma)-1}\right]\left[\displaystyle\frac{(x/\sigma)-1}{(x/\sigma)}\right]^\beta\displaystyle\frac{S(x)}{x}$\\[5mm]
\bottomrule
\end{tabular}
\end{table}

On the other hand, we can consider some extensions of these models. These extensions can be obtained using the Pareto types II or IV models \cite{Arnold2015,Arnold2008,Arnold2014}, instead of the usual classical Pareto distribution. In these extensions, we incorporate a new location parameter or new location and shape parameters, respectively, and the support of the distribution is $(\mu,\infty)$,where $\mu\ge 0$. In this situation, the $g(\cdot)$ function is continuous, positive and differentiable on the interval $(0,\infty)$ and satisfy:
\begin{equation}\label{eq4}
\displaystyle\lim_{z \to 0^+}(1+z)^{g(z)}=1\;\mbox{and}\;\displaystyle\lim_{z \to\infty}(1+z)^{g(z)}=\infty,
\end{equation}
and
\begin{equation}\label{eq5}
\displaystyle\frac{g'(z)}{g(z)}>\frac{-1}{(1+z)\log(1+z)},\;\forall\;z>0.\\[2ex]
\end{equation}

In the case of Pareto II distribution, the new family of distributions is defined in terms of the cdf,
$$F(x;\mu,\sigma)=1-\left[1+\left(\displaystyle\frac{x-\mu}{\sigma}\right)\right]^{-g\{(x-\mu)/\sigma\}},\;x>\mu,$$
and $F(x)=0$ if $x\leq\mu$, where $\mu$ is a location parameter, $\sigma>0$ is a scale parameter and $g(\cdot)$ satisfies conditions (\ref{eq4}) and (\ref{eq5}).

In the case of the Pareto IV distribution, the new family of  distributions is given by,
$$F(x;\mu,\sigma,\gamma)=1-\left[1+\left(\displaystyle\frac{x-\mu}{\sigma}\right)^{\displaystyle(1/\gamma)}\right]^{-g\{[(x-\mu)/\sigma]^{1/\gamma}\}},\;x>\mu,$$
and $F(x)=0$ if $x\leq\mu$, where $\mu$ is a location parameter, $\sigma>0$ is a scale parameter, $\gamma>0$ a shape parameter and $g(\cdot)$ satisfies again conditions (\ref{eq4}),(\ref{eq5}).

\section{A relevant model}\label{relevant}

In this section we study a relevant model of Generalized Power Law distribution. This model corresponds to the choice $g(z)=\alpha\left(\frac{\log z}{1+\log z}\right)^\beta$ in Table \ref{distexamples}, and the cdf is given by,
\begin{equation}\label{gplcdf}
F(x;\alpha,\beta,\sigma)=1-\exp\left\{-\alpha\displaystyle\frac{[\log(x/\sigma)]^{\beta+1}}{[\log(x/\sigma)+1]^\beta}\right\},\;x\geq\sigma,
\end{equation}
and $F(x)=0$ if $x<\sigma$, where $\alpha>0$ and $\beta>-1$ are shape parameters, and $\sigma>0$ is a scale parameter. A random variable with cdf given by Eq.(\ref{gplcdf}) will be denoted by $X\sim {\cal GPL}(\alpha,\beta,\sigma)$. This family includes the classical Pareto distribution (Power Law) when $\beta=0$. We have ${\cal GPL}(\alpha,0,\sigma)\equiv{\cal P}a(\alpha,\sigma)$.

The concept of tail equivalent (see \cite{Embrechts,Focardi,Sornette,Resnick2007,Castillo2012,Klugman,LeCourtois}) is satisfied by the ${\cal GPL}(\alpha,\beta,\sigma)$ distribution. 
The following Theorem shows that the ${\cal GPL}(\alpha,\beta,\sigma)$ distribution exhibit a power law behaviour when $x$ is large.
\begin{theorem}
The ${\cal GPL}(\alpha,\beta,\sigma)$ distribution, defined in Eq.(\ref{gplcdf}), and the Pareto distribution are right tail equivalent
\end{theorem}
\begin{proof}
The proof is direct and it is based on the fact that,
$$\displaystyle\lim_{x \to\infty}\displaystyle\frac{1-F(x)}{1-G(x)}=
\displaystyle\lim_{x \to\infty}\frac{\exp\left\{-\alpha\frac{[\log(x/\sigma)]^{\beta+1}}{[\log(x/\sigma)+1]^\beta}\right\}}{(x/\sigma)^{-\alpha}}=1,$$
where $G(x)$ is the cumulative distribution function of the Pareto distribution.
\end{proof}

In the following Theorem we show the domain of attraction for maxima (see \cite{Clauset,Castillo1989,Bassi,Asimit,Cirillo,Jayakrishnan,Gorge,Resnick2013})
\begin{theorem}
The ${\cal GPL}(\alpha,\beta,\sigma)$ distribution belongs to the Maximum Domain of Attraction of the Fr\'echet distribution
${\cal GPL}(\alpha,\beta,\sigma)\in MDA(\Phi_{\alpha})$
\end{theorem}
\begin{proof}
We must check that $1-F(x)$ is of regular variation of index $-\alpha$,
\begin{equation*}
\displaystyle\lim_{x \to\infty}\displaystyle\frac{1-F(tx)}{1-F(x)}=t^{-\alpha},\forall t>0,
\end{equation*}
or in other words, $1-F(x)$ can be expressed as $L(x)x^{-\alpha}$ where  $L(x)$ is a slowly varying function ($\displaystyle\lim_{x \to\infty}L(tx)/L(x)=1$, for any  $t>0$):
\begin{equation*}
1-F(x)=(x/\sigma)^{-\alpha\left\{\left[\frac{\log(x/\sigma)}{\log(x/\sigma)+1}\right]^{\beta}-1\right\}}(x/\sigma)^{-\alpha}\sim L(x)x^{-\alpha},
\end{equation*}
which means that ${\cal GPL}(\alpha,\beta,\sigma)$ is a heavy-tailed distribution and, for that, it can be useful for statistical modeling of phenomena with extremely large observations.
\end{proof}

\subsection{Basic Properties}

The survival function $S(x)=\Pr(X>x)=1-F(x)$ is given by:
\begin{equation}\label{gplsf}
S(x)=\exp\left\{-\alpha\displaystyle\frac{[\log(x/\sigma)]^{\beta+1}}{[\log(x/\sigma)+1]^\beta}\right\},\;x\geq\sigma,
\end{equation}
and $S(x)=1$ if $x<\sigma$. Figure \ref{figsf} shows the survival function $S(x)$ of the ${\cal GPL}(\alpha,\beta,\sigma)$ distribution, given by Eq.(\ref{gplsf}), for different values of the shape parameters
$\alpha$ and $\beta$, in log-log scale.

\begin{figure}[p]
\renewcommand{\figurename}{\footnotesize{Figure}}
\begin{center}
\includegraphics[scale=0.67]{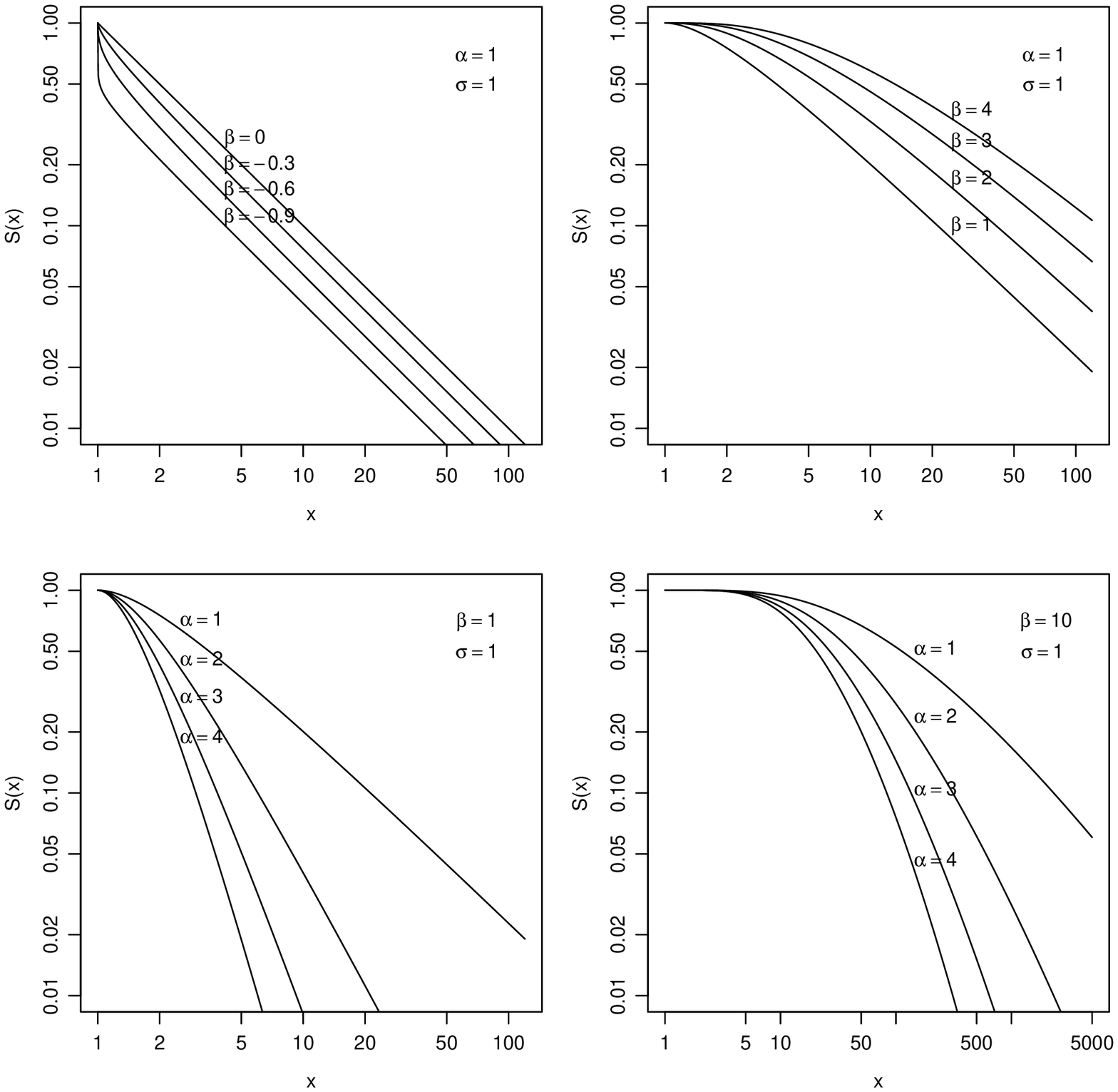}
\caption{\footnotesize{Plots of the survival function of the ${\cal GPL}(\alpha,\beta,\sigma)$ distribution with $\sigma=1$ and: (up-left) $\alpha=1$ and $\beta=-0.9,-0.6,-0.3,0$;
(up-right) $\alpha=1$ and $\beta=1,2,3,4$; (down-left) $\beta=1$ and $\alpha=1,2,3,4$; (down-right) $\beta=10$ and $\alpha=1,2,3,4$}}\label{figsf}
\end{center}
\end{figure}

The pdf of the ${\cal GPL}(\alpha,\beta,\sigma)$ distribution is given by,
\begin{equation}\label{gplpdf}
f(x)=\displaystyle\frac{\alpha}{x}\left[\displaystyle\frac{\log(x/\sigma)+1+\beta}{\log(x/\sigma)+1}\right]\left[\displaystyle\frac{\log(x/\sigma)}{\log(x/\sigma)+1}\right]^\beta
\exp\left\{-\alpha\displaystyle\frac{[\log(x/\sigma)]^{\beta+1}}{[\log(x/\sigma)+1]^\beta}\right\},\;x\geq\sigma
\end{equation}
and $f(x)=0$ if $x<\sigma$. Figure \ref{figpdf} shows the probability density function $f(x)$, given by Eq.(\ref{gplpdf}), for zero-modal and uni-modal curves.
Remark that ${\cal GPL}(\alpha,\beta,\sigma)$ distribution, as a distribution of the MDA  of the Fr\'echet distribution, satisfies the von Mises condition \cite{vonMises}:
$$\displaystyle\lim_{x \to\infty}\displaystyle\frac{x f(x)}{S(x)}=
\displaystyle\lim_{x \to\infty}
\alpha\left[\displaystyle\frac{\log(x/\sigma)+1+\beta}{\log(x/\sigma)+1}\right]\left[\displaystyle\frac{\log(x/\sigma)}{\log(x/\sigma)+1}\right]^\beta=\alpha>0.$$

\begin{figure}[p]
\renewcommand{\figurename}{\footnotesize{Figure}}
\begin{center}
\includegraphics[scale=0.67]{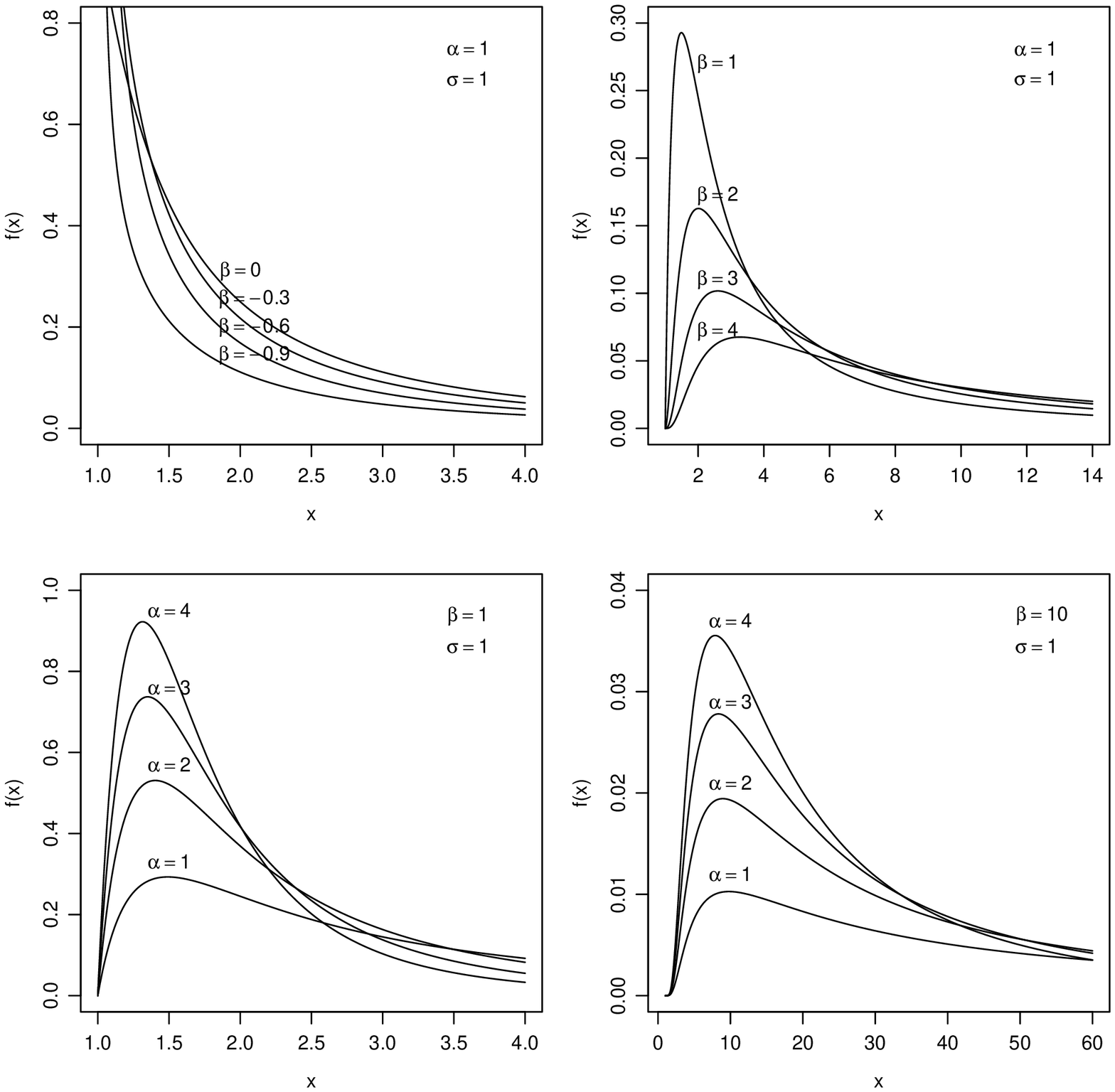}
\caption{\footnotesize{Plots of the probability density function of the ${\cal GPL}(\alpha,\beta,\sigma)$ distribution with $\sigma=1$ and: (up-left) $\alpha=1$ and $\beta=-0.9,-0.6,-0.3,0$;
(up-right) $\alpha=1$ and $\beta=1,2,3,4$; (down-left) $\beta=1$ and $\alpha=1,2,3,4$; (down-right) $\beta=10$ and $\alpha=1,2,3,4$}}\label{figpdf}
\end{center}
\end{figure}

The hazard function is given by (see also Section \ref{basics}):
\begin{equation*}
h(x)=\frac{f(x)}{\Pr(X>x)}=\frac{f(x)}{S(x)}
=\displaystyle\frac{\alpha}{x}\left[\displaystyle\frac{\log(x/\sigma)+1+\beta}{\log(x/\sigma)+1}\right]\left[\displaystyle\frac{\log(x/\sigma)}{\log(x/\sigma)+1}\right]^\beta,\;x\geq\sigma,
\end{equation*}
and $h(x)=0$ if $x<\sigma$. See Figure \ref{fighx} for different shapes.

\begin{figure}[t]
\renewcommand{\figurename}{\footnotesize{Figure}}
\begin{center}
\includegraphics[scale=0.67]{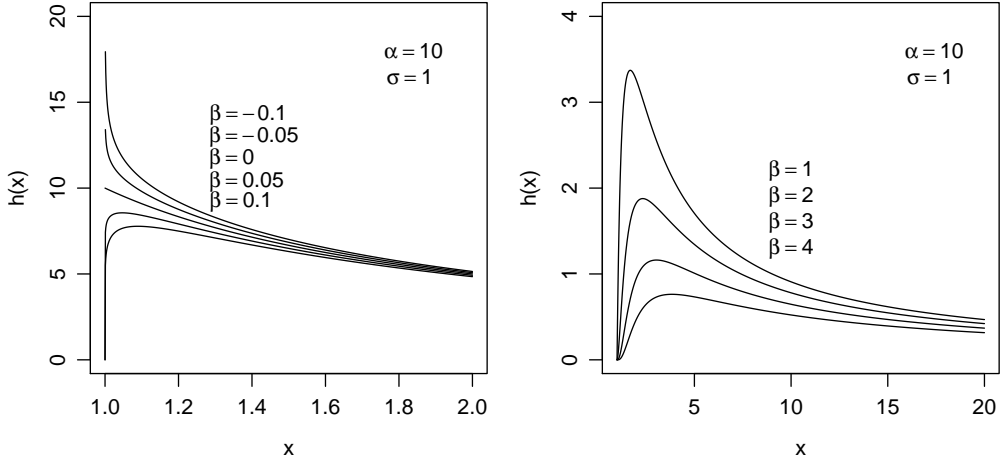}
\caption{\footnotesize{Plots of the hazard function of the ${\cal GPL}(\alpha,\beta,\sigma)$ distribution with $\sigma=1$ and: (left) $\alpha=10$ and $\beta=-0.1,-0.05,0,0.05,0.1$;
(right) $\alpha=10$ and $\beta=1,2,3,4$}}\label{fighx}
\end{center}
\end{figure}

The quantile function $Q(p)=F^{-1}(p)$ is defined implicitly as follows,
\begin{equation*}\label{quantile}
\displaystyle\frac{-\log(1-p)\left(\log\left[Q(p)/\sigma\right]+1\right)^\beta}{\left(\log\left[Q(p)/\sigma\right]\right)^{\beta+1}}=\alpha,\;\;0<p<1,
\end{equation*}
which can be used to simulate the random variable $X\sim{\cal GPL}(\alpha,\beta,\sigma)$ of our interest by the inverse transform method, from a random variable uniformly distributed $U\sim{\cal U}(0,1)$ and
$X=F^{-1}(U)$ \cite{Tyszer2012,Glasserman2003}.

\subsection{Estimation and Testing}

The ${\cal GPL}(\alpha,\beta,\sigma)$ distribution can be fitted using the method of maximum likelihood \cite{Fisher1922}.
Let $x_1,\dots, x_n$ be a sample of size $n$ drawn from a ${\cal GPL}(\alpha,\beta,\sigma)$ distribution.
The log-likelihood function can be expressed as follows,
\begin{equation}\label{loglikelihood}
\begin{split}
\log\ell(\alpha,\beta,\sigma)&=\sum_{i=1}^n\log f(x_i)=n\log(\alpha)-\sum_{i=1}^n\log(x_i)\\
&+\sum_{i=1}^n\log[\log(x_i/\sigma)+1+\beta]+\beta\sum_{i=1}^n\log[\log(x_i/\sigma)]\\
&-(\beta+1)\sum_{i=1}^n\log[\log(x_i/\sigma)+1]-\alpha\sum_{i=1}^n\displaystyle\frac{\log^{\beta+1}(x_i/\sigma)}{[\log(x_i/\sigma)+1]^\beta},
\end{split}
\end{equation}
where $(\alpha,\beta,\sigma)$ is the unknown parameter vector of the model, $f(x)$ is the pdf of the ${\cal GPL}(\alpha,\beta,\sigma)$ distribution defined in Eq.(\ref{gplpdf}), and the maximum
likelihood estimation of the parameter vector $(\hat{\alpha},\hat{\beta},\hat{\sigma})$ is the one that maximizes the likelihood function $\log\ell(\alpha,\beta,\sigma)$.

The normal equations can be obtained by taking partial derivatives of Eq.(\ref{loglikelihood}) with respect to $\alpha,\beta,\sigma$, and equating them to zero. If we use the auxiliar random variable
$W=\log(X/\sigma)$, and we represent its observed values by $w_i=\log(x_i/\sigma),i=1,\dots,n$, the normal equations can be expressed as:
\begin{equation}\label{partials}
\begin{gathered}
\frac{\partial\log\ell}{\partial\alpha}=0\Rightarrow \alpha=n\left[\sum_{i=1}^n\displaystyle\frac{w_i^{\beta+1}}{(w_i+1)^\beta}\right]^{-1},\\
\frac{\partial\log\ell}{\partial\beta}=0\Rightarrow
\sum_{i=1}^n
\displaystyle\frac{1}{w_i+1+\beta}
+\log\left(\displaystyle\frac{w_i}{w_i+1}\right)\left[1-\displaystyle\frac{\alpha w_i^{\beta+1}}{(w_i+1)^\beta}\right]=0,\\
\end{gathered}
\end{equation}
\begin{equation*}
\frac{\partial\log\ell}{\partial\sigma}=0\Rightarrow
\sum_{i=1}^n
\displaystyle\frac{\beta(\beta+1)[\log(w_i)+1]^\beta- \alpha w_i\log^{\beta+1} (w_i) [\log(w_i)+1+\beta]^2}
{w_i\log(w_i)[\log(w_i)+1]^{\beta+1}[\log(w_i)+1+\beta]}=0.
\end{equation*}

The previous equations (\ref{partials}) can be solved by numerical methods. In this study, maximum likelihood estimates of the parameters
$\alpha$, $\beta$ and $\sigma$ were computed by using the \verb|R| software function \verb|optimx| \cite{Rproject,Nash},
with the limited memory quasi-Newton L-BFGS-B algorithm (in which bounds contraints are permited) \cite{Byrd,Barros,Lange}
- for that, we took $\sigma_0$ equal to half of the smallest value of the sample, as the initial value of $\sigma$, and $\alpha_0,\beta_0$
the values obtained from the first two partial derivatives (in Eq.(\ref{partials})) just plugging $\sigma_0$ into them.

We can compare the ${\cal GPL}(\alpha,\beta,\sigma)$ distribution with other different models by using two model selection criteria: the Akaike information criterion ($AIC$), defined by \cite{Akaike},
$$AIC=-2\log L+2d;$$
or the Bayesian information criterion ($BIC$), defined by \cite{Schwarz}
\begin{equation}\label{bic}
BIC=\log L-\frac{1}{2}d\log n;
\end{equation}
where $\log L=\log\ell(\hat{\alpha},\hat{\beta},\hat{\sigma})$ is the log-likelihood (see Eq. \ref{loglikelihood}) of the model evaluated at the maximum likelihood estimates, $d$ is the number of parameters 
(in the case of the ${\cal GPL}(\alpha,\beta,\sigma)$ distribution, $d=3$) and $n$ is the number of data. The model chosen is the one with the smallest value of $AIC$ statistic or with the largest value of $BIC$ statistic.

We can use rank-size plots (on a log-log scale) for graphical model validation.
We can plot the complementary of the theoretical cdf (multiplied by $n + 1$) of the ${\cal GPL}(\alpha,\beta,\sigma)$ model together with the scatter plot of the points (observed data)
$\log(rank_i)$ versus $\log(x_{(i)}),\;i=1,\dots,n$, where $x_{(1)}\leq\dots\leq x_{(i)}\leq\dots\leq x_{(n)}$ is the ordered sample of $X$ and $rank_i=n+1-i$ \cite{Prieto2014}.

Finally, we can test the goodness-of-fit of the ${\cal GPL}(\alpha,\beta,\sigma)$ model by a Kolmogorov-Smirnov ($KS$) test method based on bootstrap resampling \cite{Clauset,Prieto2014,Efron,Wang,Babu,Kolmogorov,Smirnov} as follows:
(1) calculating the empirical $KS$ statistic of the ${\cal GPL}(\alpha,\beta,\sigma)$ model for the observed data,
$KS= \sup\;\lvert F_n(x_{i})-F(x_{i};\hat{\alpha},\hat{\beta},\hat{\sigma}) \rvert,\;i=1,2,\dots,n$,
where $F(x_{i};\hat{\alpha},\hat{\beta},\hat{\sigma})$ is the theoretical cdf of the ${\cal GPL}(\alpha,\beta,\sigma)$ model fitted by maximum likelihood, in a sample value, and
$F_n(x_{i})\approx(n+1)^{-1}\sum_{j=1}^n I_{[x_{j}\leq x_{i}]}$ is the empirical cdf in a sample value with the indicated plotting position formula \cite{Castillo2005};
(2) generate, by simulation, enough ${\cal GPL}(\alpha,\beta,\sigma)$ synthetic data sets (in this study, we generated 10000 data sets), with the same sample size $n$
- notice that the ${\cal GPL}(\alpha,\beta,\sigma)$ quantile function $Q(p)=F^{-1}(p)$ is defined implicitly, then, for this study, we used the \verb|R| software function \verb|uniroot|  \cite{Rproject};
(3) fit each ${\cal GPL}(\alpha,\beta,\sigma)$ synthetic data set by maximum likelihood and obtained its theoretical cdf;
(4) calculate the $KS$ statistic for each ${\cal GPL}(\alpha,\beta,\sigma)$ synthetic data set - with its own theoretical cdf;
(5) calculate the $p$-value as the fraction of ${\cal GPL}(\alpha,\beta,\sigma)$ synthetic data sets with a $KS$ statistic greater than the empirical $KS$ statistic;
(6) null hypothesis {\it $H_0$: the data follow the ${\cal GPL}(\alpha,\beta,\sigma)$ model} can be rejected with the 0.1 level of significance if $p$-value$<0.1$.

\section{Empirical application to municipal debt in Spain}\label{Application}

In this section, as an illustration, we show that ${\cal GPL}(\alpha,\beta,\sigma)$ distribution can be useful for modeling Spanish municipalities debt.

\subsection{The data}

We considered debt data of the indebted municipalities in Spain.
There are three levels of government in Spain: the State, the Autonomous Communities and the Local Entities \cite{DGCL,Benito2004}. Municipalities belong to the third one - as a reference, there were 8117 municipalities in Spain in 2014 \cite{INE}.
The expenditure of those councils, directed at providing essential local services to their citizens (street cleaning,
local police, etc.), is financed through different sources: transfers, local taxes, public fares, etc.
For several reasons (infrastructure investment, etc.), they can decide to contract debt - taking into account
the municipal debt control of the institutional borrowing restrictions  \cite{Montesinos2000,Alba2003,Sole2006,Cabases2007,Bastida2009,Hita2011,GarciaSanchez2011,GarciaSanchez2012,Lopez2012,Almendral2013,Benito2015}.
Our data sets were composed of information of Spanish indebted municipalities, whose debt was at least one thousand euros, dated on the 31st of December of each year, in the period 2008-2014,
expressed in thousand of euros ($k\euro{}$), published by the Spanish Ministry of the Finance and Public Administrations (see \cite{MFPA}).

Table \ref{data} show the main empirical characteristics of the variable of our interest: the number of Spanish indebted municipalities analyzed ($n$); the total amount of borrowing of those indebted municipalities; the debt of the most indebted council; the minimum value of debt considered; the mean and standard deviation (in $k\euro{}$); the skewness and kurtosis of that municipal debt.
\begin{table}[htbp]\scriptsize
\renewcommand{\tablename}{\footnotesize{Table}}
\caption{\label{data}\footnotesize Some relevant information about the datasets considered.}
\centering
\setlength{\tabcolsep}{1 mm}
\begin{tabular}{@{}l c c c c c c c @{}}
\toprule
Year                                                                          	& 2008           		& 2009    			& 2010 			& 2011		 	& 2012		 	 & 2013 			& 2014   \\
\midrule
Indebted Municip. ($n$)	              				& 4,981          		& 5,083	    		& 5,039     		& 4,979		 	& 5,059		 	& 5,028 			& 4,668   \\
Total Amount ($k\euro{}$)   					& $25.2\times10^6$	& $28.1\times10^6$	& $28.5\times10^6$	& $28.2\times10^6$	& $35.2\times10^6$	& $34.9\times10^6$	& $31.3\times10^6$\\
Maximum ($k\euro{}$)                      				& $6.7\times10^6$	& $6.8\times10^6$	& $6.5\times10^6$	& $6.3\times10^6$	& $7.4\times10^6$	& $7.0\times10^6$	& $5.9\times10^6$\\
Minimum  ($k\euro{}$)     						& 1				& 1				& 1				& 1				& 1				& 1				& 1.78		\\
Mean       ($k\euro{}$)     	 		     			& 5,059.2			& 5,532.2 	 		& 5,649.1	 		& 5,655.6 	 		& 6,950.6 	 		& 6,942.4	 		 & 6,715.4 	 \\
Std. Dev. ($k\euro{}$)     	     					& 97,714.1		& 98,603.0		& 95,643.6		& 94,555.1		& 109,764.6		& 104,657.1		& 92,645.3  \\
Skewness        			    	 			     	& 64.5			& 64.1			& 61.6			& 61.3			& 61.7			& 60.9			& 57.0  \\
Kurtosis 			    			 		     	& 4,384.5			& 4,378.4 			& 4,105.7 			& 4,072.0 			& 4,138.9 			& 4,051.3 			 & 3,604.9   \\
\bottomrule
\end{tabular}
\end{table}

\subsection{Power Law behavior in the upper tail}

We analyzed the power law behavior of the Spanish municipal debt. For that, we followed the methodology proposed in Clauset et al. \cite{Clauset2009}, based on:
(1) the maximum likelihood method, for fitting the Pareto distribution to the data - in this case, the maximum likelihood estimator for the scale parameter is the minimum value of the sample: $\hat{\sigma}=x_{min}$,
and the maximum likelihood estimator for the shape parameter $\hat{\alpha}$ is the Hill estimator \cite{Hill1975} given by
\begin{equation*}\label{aic}
\hat{\alpha}=n\left[\sum_{i=1}^n\log(x_i/x_{min})\right]^{-1};
\end{equation*}
(2) the Kolmogorov-Smirnov ($KS$) test method based on bootstrap resampling, for testing the goodness-of-fit of Pareto model;
and (3), for estimating the lower bound $x_{min}$ of the power law behavior, an iterative algorithm where  $x_{min}$ is given by the minimum sample value in which the null hypothesis {\it $H_0$: the data follow a power law model} can't be rejected at 0.1 level of significance.

Table \ref{respl} shows, for each year: the shape parameter estimates $\hat{\alpha}$ obtained from the datasets analyzed; the corresponding scale parameter estimates $\hat{\sigma}$, which give us the minimum local debt that follows the power law behavior (in $k\euro{}$); the number of municipalities that follow that behavior; the empirical $KS$ statistics and the $p$-values obtained. It can be seen that power law behavior is only valid in the upper tail of the distribution - only the largest debts can be modeled with a classical Pareto distribution - since null hypothesis {\it $H_0$: the data follow a power law model} can be rejected at the 0.1 level of significance for values of $x_{min}$ less than $\hat{\sigma}$ and, in particular, it can be rejected if we considered the whole range of the distribution.

In addition, table \ref{respl} shows that shape parameter estimates $\hat{\alpha}$ are very close to 1
(Zipf's law for many authors, \cite{Brakman,Gabaix1999a,Gabaix1999b,Urzua,Ioannides,Fujiwara,Gabaix2004,Anderson,Cordoba})
for the first four years analyzed (2008-2011), and that they change in 2012 (likewise, $\hat{\sigma}$ and $n$)
- coinciding with the political scene change after the Spanish municipal, regional and general elections held on 2011.

\begin{table}[h]\scriptsize
\renewcommand{\tablename}{\footnotesize{Table}}
\caption{\label{respl}\footnotesize Parameter estimates ($\hat{\alpha},\hat{\sigma}$)  from the Power Law model to the upper tail of the local debt datasets by maximum likelihood; number of municipalities ($n$) with the largest debts, which follow a power law behaviour; empirical $KS$ statistics; and bootstrap $p$-values for that model (values of $p < 0.1$ indicate that the models can be ruled out with the 0.1 level of significance).}
\centering
\setlength{\tabcolsep}{1.8 mm}
\begin{tabular}{@{}l c c c c c c c @{}}
\toprule		
Year                                                                         			 & 2008           	& 2009    		& 2010 		& 2011 		& 2012 		& 2013 		 & 2014   \\
\midrule	
$\hat{\alpha}$: shape parameter estimates	              		& 0.9981      	& 1.0116    	& 0.9990     	& 1.0207 		& 0.8557 		& 0.8455 		& 0.8322   \\
$\hat{\sigma}$: lower bound ($k\euro{}$) , scale par. estim.  	& 9253         	& 10582 		& 10692   		& 12563  		& 4677		& 5014 		& 4101 \\
$n$: size (municipalities) of the upper tail	                  		& 343          	& 348   		& 345     		& 298    		& 760 		& 700 		& 742   \\
Empirical $KS$ statistics						                 & 0.0513    	& 0.0505   	& 0.0505     	& 0.0529    	& 0.0350 		& 0.0364 		& 0.0351   \\
$p$-value 	($>0.1$ favor power law model)				& 0.1037 	  	& 0.1100 	  	& 0.1114 	  	& 0.1293 	  	& 0.1056 	  	& 0.1050 	  	& 0.1084  \\
\bottomrule
\end{tabular}
\end{table}

Figure \ref{fig04} shows, as a graphical model validation, the rank-size plots (on log-log scale) in the selected years 2008, 2011 and 2014, for the
whole range of the datasets (left) and for the upper tail of the distribution (right). Those plots confirm, graphically, that power law model can be ruled out as an adequate model for the whole range of indebted municipalities, and that power law model may serve as an adequate model for municipalities with largest debts above a certain lower bound, in accordance with Table \ref{respl}.

\begin{figure}[p]
\renewcommand{\figurename}{\footnotesize{Figure}}
\begin{center}
\includegraphics*[width=1.0\textwidth]{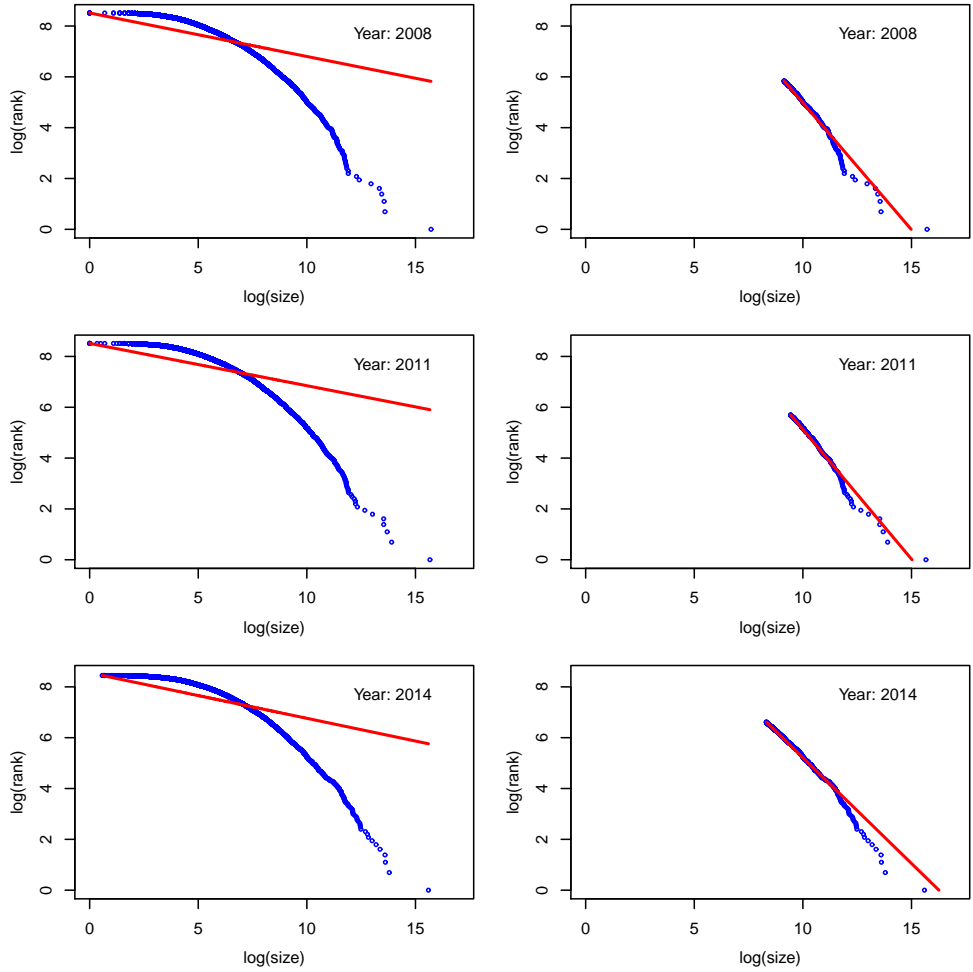}
\caption{\label{fig04}\footnotesize{Rank-size plots of the complementary of the cdf multiplied by $n+1$ (solid lines) of the classical Pareto distribution (power law model) and the observed data, on log-log scale. Left: Whole range. Right: Upper Tail. Data: Debt of the Spanish indebted municipalities in 2008, 2011 and 2014, whose debt was at least one thousand euros, dated on the 31st of December of each year, in thousand of euros, published by the Spanish Ministry of the Finance and Public Administrations.}}
\end{center}
\end{figure}

\subsection{The new distribution in the whole range}

In this section, we compare the ${\cal GPL}(\alpha,\beta,\sigma)$ distribution with other eight models, and we test the adequacy of the ${\cal GPL}(\alpha,\beta,\sigma)$ distribution to the datasets in the whole range. 

We fitted the ${\cal GPL}(\alpha,\beta,\sigma)$ model and those eight known models to the datasets, in the whole range, by maximum likelihood, from 2008 to 2014.
Four of those eight models with two parameters:
Pareto (Power Law); Lomax (Pareto type II with location parameter $\mu=0$) \cite{Lomax1954}; Lognormal \cite{Johnson1994} and Fisk (Log-logistic) \cite{Fisk1961} distributions.
The other four models with three parameters: Pareto type II; three-parameter lognormal; Burr type XII (Singh-Maddala) \cite{Burr1942,Singh1976} and Dagum \cite{Dagum1975} distributions.
Table \ref{distributions} shows the cumulative distribution functions $F(x)$ and the probability density functions $f(x)$ of the nine models considered.
\begin{table}[p]\tiny
\renewcommand{\tablename}{\footnotesize{Table}}
\caption{\label{distributions}\footnotesize Cumulative distribution functions and probability density functions of the models fitted to the dataset in the whole range. $\Phi(z)$ denotes the standard normal CDF.}
\centering
\setlength{\tabcolsep}{2.4 mm}
\begin{tabular}{@{}l c c @{}}
\toprule		
Distribution  & $F(x)$  & $f(x)$   \\
\midrule
Pareto
&	
$1-\left(\displaystyle\frac{x}{\sigma}\right)^{-\alpha}$
&	
$\displaystyle\frac{\alpha \sigma^\alpha}{x^{\alpha+1}}\;x\geq\sigma$
\\[2ex]
Lomax
&	
$1-\left(1+\displaystyle\frac{x}{\sigma}\right)^{-\alpha}$
&	
$\displaystyle\frac{\alpha \sigma^\alpha}{(x+\sigma)^{\alpha+1}}\;x\geq 0$	
\\[2ex]
Lognormal
&	
$\Phi\left(\displaystyle\frac{\log x-\mu}{\sigma}\right)$
&
$\displaystyle\frac{1}{x\sigma\sqrt{2\pi}}\exp\left[-\frac{(\log x-\mu)^2}{2\sigma^2}\right]\;x>0$
\\[2ex]
Fisk
&	
$\displaystyle\frac{1}{1+(x/\alpha)^{-\beta}}$
&	
$\displaystyle\frac{(\beta/\alpha)(x/\alpha)^{\beta-1}}{(1+(x/\alpha)^{\beta})^2}\;x>0$
\\[2ex]
Pareto II
&
$1-\left(1+\displaystyle\frac{x-\mu}{\sigma}\right)^{-\alpha}$
&
$\displaystyle\frac{\alpha \sigma^\alpha}{(x-\mu+\sigma)^{\alpha+1}}\;x\geq \mu$	
\\[2ex]
Lognormal 3p
&	
$\Phi\left(\displaystyle\frac{\log(x-\gamma)-\mu}{\sigma}\right)$
&
$\displaystyle\frac{1}{\sigma(x-\gamma)\sqrt{2\pi}}\exp\left[-\frac{(\log(x-\gamma)-\mu)^2}{2\sigma^2}\right]\;x>\gamma$
\\[2ex]
Burr type XII
&	
$1-\left[1+{\left(\displaystyle\frac{x}{b}\right)}^a\right]^{-q}$
&
$\displaystyle\frac{aqx^{a-1}}{b^{a}[1+(x/b)^a]^{q+1}},\;x\geq 0$
\\[2ex]
Dagum
&	
$\left[1+{\left(\displaystyle\frac{x}{b}\right)}^{-a}\right]^{-p}$
&
$\displaystyle\frac{apx^{ap-1}}{b^{ap}[1+(x/b)^a]^{p+1}},\;x\geq 0$
\\[2ex]
${\cal GPL}(\alpha,\beta,\sigma)$ 
&	
$1-\exp\left\{-\alpha\displaystyle\frac{[\log(x/\sigma)]^{\beta+1}}{[\log(x/\sigma)+1]^\beta}\right\}$
&	
$\displaystyle\frac{\alpha}{x}\left[\displaystyle\frac{\log(x/\sigma)+1+\beta}{\log(x/\sigma)+1}\right]\left[\displaystyle\frac{\log(x/\sigma)}{\log(x/\sigma)+1}\right]^\beta
\exp\left\{-\alpha\displaystyle\frac{[\log(x/\sigma)]^{\beta+1}}{[\log(x/\sigma)+1]^\beta}\right\}$
\\
\bottomrule
\end{tabular}
\end{table}

We compared those models using the Bayesian information criterion ($BIC$, see Eq.(\ref{bic})).
Table \ref{bicst} shows the $BIC$ statistics obtained, from the nine selected models (ranked by $BIC$), corresponding to our datasets in the whole range, from 2008 to 2014.
${\cal GPL}(\alpha,\beta,\sigma)$ distribution presents the largest values of $BIC$ statistics, therefore ${\cal GPL}(\alpha,\beta,\sigma)$ distribution is the model chosen using that model selection criterion.
Table \ref{GPLparameters} shows the corresponding parameter estimates and their standard errors from the ${\cal GPL}(\alpha,\beta,\sigma)$ distribution.

We checked graphically the adequacy of the ${\cal GPL}(\alpha,\beta,\sigma)$ distribution to the datasets in the whole range using rank-size plots. Figure \ref{fig05} shows the plots obtained from 2008 to 2014.

Finally, we tested the goodness-of-fit of ${\cal GPL}(\alpha,\beta,\sigma)$ distribution, by a Kolmogorov-Smirnov ($KS$) test method based on bootstrap resampling. Table \ref{KSpvalue} shows the values of the empirical
$KS$ statistics and the $p$-values obtained. It can be seen that we obtained $p$-values $\geq0.1$ in three of the seven years considered.

In summary, ${\cal GPL}(\alpha,\beta,\sigma)$ distribution can be useful for modeling Spanish municipalities debt: it presents the best $BIC$ statistics of the nine selected models;
graphically it gives a reasonable description of the datasets; and it cannot be rejected with 0.1 level of significance in three of the seven years considered.

\begin{table}[p]\scriptsize
\renewcommand{\tablename}{\footnotesize{Table}}
\caption{\label{bicst}\footnotesize $BIC$ statistics for nine candidate models, fitted by maximum likelihood to municipal debt data in Spain. Larger values indicate better fitted models (models appear ranked by $BIC$).}
\centering
\setlength{\tabcolsep}{4.2 mm}
\begin{tabular}{@{}l c c c c c c c @{}}
\toprule		
Year                             		 	& 2008           	& 2009    		& 2010 		& 2011 		& 2012 		& 2013 		& 2014   \\
\midrule	
${\cal GPL}(\alpha,\beta,\sigma)$	& -39656         	& -41041 		& -40892   	& -40528  		& -42316		& -41925		& -38840 \\
Lognormal 3p					& -39723         	& -41092 		& -40954   	& -40579  		& -42368		& -41974 		& -38859 \\
Dagum			  			& -39708         	& -41093 		& -40955   	& -40578  		& -42380		& -41985 		& -38884 \\
Lognormal	              			& -39733      	& -41101    	& -40966     	& -40587 		& -42378 		& -41984 		& -38877   \\
Pareto II	              				& -39731      	& -41119    	& -40986     	& -40607 		& -42417 		& -42015 		& -38893   \\
Lomax			  			& -39748        	& -41135 		& -41000   	& -40620  		& -42427		& -42026 		& -38914   \\
Burr type XII	              			& -39746      	& -41139    	& -41003     	& -40624 		& -42432 		& -42030		& -38917  \\
Fisk			  				& -39792        	& -41172 		& -41036   	& -40652  		& -42448		& -42046 		& -38927    \\
Pareto	              				& -42870		& -44251   	& -44200    	& -43820 		& -45819 		& -45338 		& -41486   \\
\bottomrule
\end{tabular}
\end{table}

\begin{table}[p]\scriptsize
\renewcommand{\tablename}{\footnotesize{Table}}
\caption{\label{GPLparameters}\footnotesize Parameter estimates from the ${\cal GPL}(\alpha,\beta,\sigma)$ model, to Spanish municipal debt datasets, in the whole range,
by maximum likelihood (standard errors in parenthesis).}
\centering
\setlength{\tabcolsep}{4 mm}
\begin{tabular}{@{}l c c c c c c c @{}}
\toprule		
Year                                   & 2008           	& 2009    		& 2010 		& 2011 		& 2012 		& 2013 		& 2014   \\
\midrule	
$\hat{\alpha}$	              	& 2.3477      	& 2.6126    	& 2.3085     	& 2.6703 		& 2.5098 		& 2.5722 		& 3.0427   \\
					& (0.1910)      	& (0.2336)    	& (0.1979)     	& (0.2419) 	& (0.2387) 	& (0.2445)		& (0.3325)   \\
$\hat{\beta}$			& 24.3346         & 27.4832 	& 24.5065   	& 27.1347  	& 26.1137		& 26.9485 	& 30.5616 \\
					& (1.3320)      	& (1.5758)    	& (1.4319)     	& (1.5735) 	& (1.6499) 	& (1.6755) 	& (2.0480)   \\
$\hat{\sigma}$			& 0.2367         	& 0.1536 		& 0.2536   	& 0.1882  		& 0.2625		& 0.2199		& 0.1335 \\
					& (0.0392)      	& (0.0291)    	& (0.0454)     	& (0.0355) 	& (0.0530) 	& (0.0447) 	& (0.0316)   \\
\bottomrule
\end{tabular}
\end{table}

\begin{figure}[p]
\renewcommand{\figurename}{\footnotesize{Figure}}
\begin{center}
\includegraphics[width=1.0\textwidth]{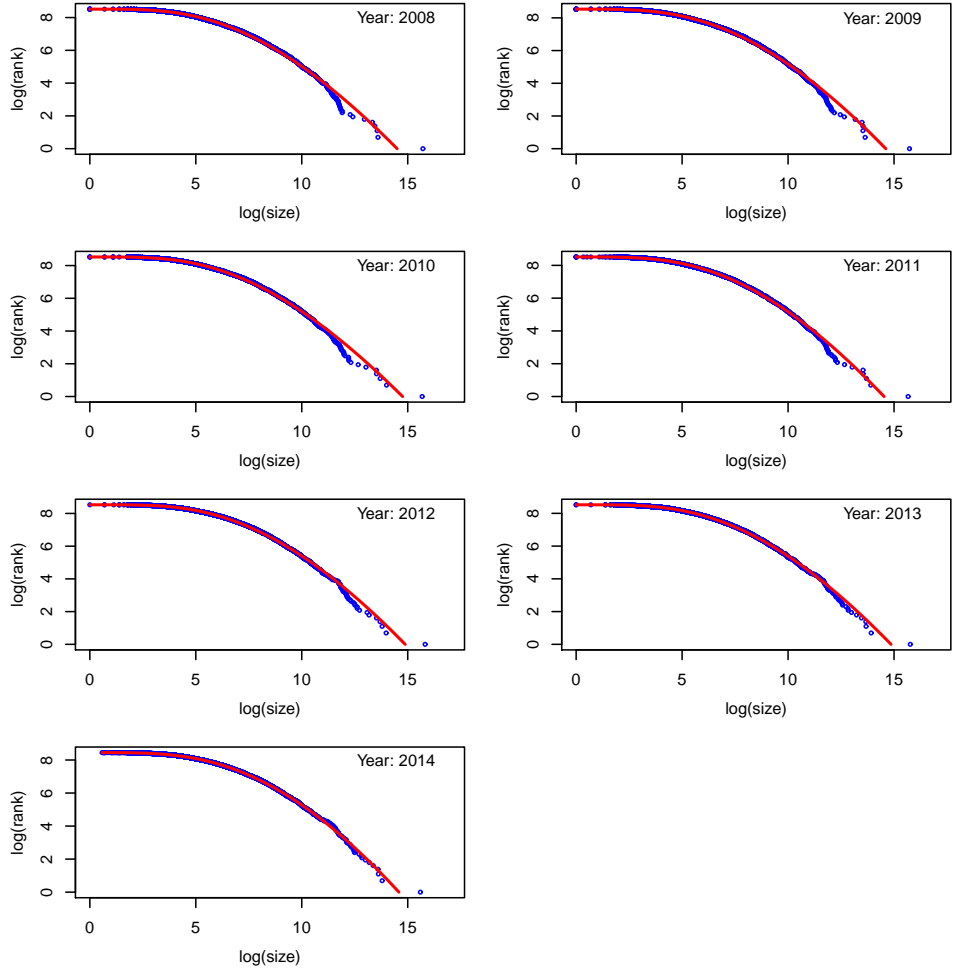}
\caption{\label{fig05}\footnotesize{Rank-size plots of the complementary of the cdf multiplied by n + 1 (solid lines) of the ${\cal GPL}(\alpha,\beta,\sigma)$ distribution and the observed data, on log-log scale.
Data: Debt of the Spanish indebted municipalities, from 2008 to 2014, whose debt was at least one thousand euros, dated on the 31st of December of each year, in thousand of euros,
published by the Spanish Ministry of the Finance and Public Administrations.}}
\end{center}
\end{figure}

\begin{table}[t]\scriptsize
\renewcommand{\tablename}{\footnotesize{Table}}
\caption{\label{KSpvalue}\footnotesize Empirical KS statistics and bootstrap $p$-values for ${\cal GPL}(\alpha,\beta,\sigma)$ model  ($p\geq0.1$ favor GPL model)}
\centering
\setlength{\tabcolsep}{5 mm}
\begin{tabular}{@{}l c c c c c c c @{}}
\toprule		
Year                                   & 2008           	& 2009    		& 2010 		& 2011 		& 2012 		& 2013 		& 2014   \\
\midrule	
$KS$	              		& 0.0180      	& 0.0181    	& 0.0139     	& 0.0144 		& 0.0144 		& 0.0098 		& 0.0093 \\
$p$-value			  	& 0.0000         	& 0.0000 		& 0.1000 		& 0.0048  		& 0.0056		& 0.1977 		& 0.3258 \\
\bottomrule
\end{tabular}
\end{table}

\section{Conclusions}\label{conclusions}

In this study, we focussed on modeling the whole range of empirical data, whose upper tail follows a power-law behaviour. 
To do that, we modeled the exponent of the classical Pareto distribution using a non-linear function.

We found a new family of Generalized Power Law distributions, with three parameters, which includes Pareto (power law, Pareto type I) and PPS distributions as special cases.
We showed that it is a genuine family of distributions. We provided some particular functional forms of that family.
And, as an extension, we presented two more new families of distributions based on Pareto type II and Pareto type IV distributions respectively.

We found a new distribution, the ${\cal GPL}(\alpha,\beta,\sigma)$ distribution, which belongs to the new family of Generalized Power Law distributions described previously. We showed that ${\cal GPL}(\alpha,\beta,\sigma)$ and Power Law models are right tail equivalent (${\cal GPL}(\alpha,\beta,\sigma)$ model exhibit a power law behavior in the tail) and that ${\cal GPL}(\alpha,\beta,\sigma)$ model belongs to the Maximum Domain of Attraction of the Fréchet distribution, which means that ${\cal GPL}(\alpha,\beta,\sigma)$ is a heavy-tailed distribution and it can be useful for statistical modeling of real phenomena with extremely large observations.
We provided the genesis, the basic properties (including the quantile function for computer simulation), and the corresponding estimation and testing methods for that distribution.

Finally, we provided empirical evidence of the efficacy of the ${\cal GPL}(\alpha,\beta,\sigma)$ distribution (and for extension, of the new family of Generalized Power Law distribution) with real datasets in the whole range. In particular, we showed that ${\cal GPL}(\alpha,\beta,\sigma)$ model can be useful for modeling municipal debt data. For that, we considered information of Spanish indebted municipalities, whose debt was at least one thousand euros, in the period 2008-2014, published by the Spanish Ministry of the Finance and Public Administrations. We showed that the Spanish
municipal debt follows a power law in the tail but not in the whole range. And finally, we showed analytically and graphically the competence of the ${\cal GPL}(\alpha,\beta,\sigma)$ distribution with
municipal debt data in the whole range, in comparison with other known distributions as the Lognormal, the Generalized Pareto, the Fisk, the Burr type XII and the Dagum distributions.

\section*{Acknowledgements}
The authors gratefully acknowledge financial support from the Programa Estatal de Fomento de la Investigaci\'on Cient\'ifica y T\'ecnica de Excelencia/Spanish Ministry of Economy and Competitiveness. 
Ref. ECO2013-48326-C2-2-P. In addition, this work is part of the Research Project APIE 1/2015-17: "New methods for the empirical analysis of financial markets" of the Santander Financial Institute (SANFI) of UCEIF Foundation resolved by the University of Cantabria and funded with sponsorship from Banco Santander.

\bibliographystyle{plain}

\begin{thebibliography}{11}

\bibitem{Pinto}
Pinto CM, Lopes AM, Machado JT. A review of power laws in real life phenomena. Communications in Nonlinear Science and Numerical Simulation 2012; 17(9): 3558-3578.

\bibitem{Clauset}
Clauset A, Shalizi CR, Newman ME. Power-law distributions in empirical data. SIAM review 2009; 51(4): 661-703.

\bibitem{Newman2005}
Newman ME. Power laws, Pareto distributions and Zipf's law. Contemporary physics 2005; 46(5): 323-351.

\bibitem{Clementi2006}
Clementi F, Di Matteo T, Gallegati M. The power-law tail exponent of income distributions. Physica A: Statistical Mechanics and its Applications 2006; 370(1): 49-53.

\bibitem{Rosas}
Rosas-Casals M, Solé R. Analysis of major failures in Europe's power grid. International Journal of Electrical Power \& Energy Systems 2011; 33(3): 805-808.

\bibitem{Mayo}
Mayo DG, Cox DR. Frequentist statistics as a theory of inductive inference. Lecture Notes-Monograph Series 2006; 77-97.

\bibitem{Kelton2000}
Kelton WD, Averill ML. Simulation modeling and analysis. Boston: McGraw Hill, 2000.

\bibitem{Prieto2014}
Prieto F, Sarabia JM, Sáez AJ. Modelling major failures in power grids in the whole range. International Journal of Electrical Power \& Energy Systems 2014; 54: 10-16.

\bibitem{Cuadra2015}
Cuadra L, Salcedo-Sanz S, Del Ser J, Jiménez-Fernández S, Geem ZW. A Critical Review of Robustness in Power Grids Using Complex Networks Concepts.
Energies 2015; 8(9): 9211-9265.

\bibitem{Stumpf2005}
Stumpf MP, Ingram PJ, Nouvel I, Wiuf C. Statistical model selection methods applied to biological networks. In: Transactions on Computational Systems Biology III (pp. 65-77). Springer Berlin Heidelberg; 2005.

\bibitem{Arnold1983}
Arnold BC, Pareto Distributions. International Co-operative Publishing House, Fairland, Maryland; 1983.

\bibitem{Arnold2015}
Arnold BC. Pareto distributions, second edition. In: Monographs on statistics and applied probability 2015; 140. CRC Press. Taylor \& Francis Group.

\bibitem{Sarabia2009}
Sarabia JM, Prieto F. The Pareto-Positive stable distribution: A new descriptive model for city size data. Physica A: Statistical Mechanics and its Applications 2009; 388(19): 4179-4191.

\bibitem{Pareto1964}
Pareto V. Cours d'économie politique. Librairie Droz, 1964.

\bibitem{Guillen2011}
Guill\'en M, Prieto F, Sarabia JM. Modelling losses and locating the tail with the Pareto Positive Stable distribution. Insurance: Mathematics and Economics  2011; 3(49): 454-461.

\bibitem{Arnold2008}
Arnold BC. Pareto and generalized pareto distributions. In: Modeling income distributions and Lorenz curves 2008; 119-145. Springer New York.

\bibitem{Arnold2014}
Arnold BC. Univariate and multivariate Pareto models. Journal of Statistical Distributions and Applications 2014; 1(1):1-16.

\bibitem{Embrechts}
Embrechts P, Kl\"uppelberg C, Mikosch T. Modelling extremal events,  vol. 33. Springer Science \& Business Media; 1997.

\bibitem{Focardi}
Focardi SM, Fabozzi FJ. The mathematics of financial modeling and investment management, vol. 138. John Wiley \& Sons; 2004.

\bibitem{Sornette}
Sornette D. Critical Phenomena in natural sciences: chaos, fractals, selforganization and disorder: concepts and tools. Springer Series in Synergetics; 2006.

\bibitem{Resnick2007}
Resnick SI. Heavy-tail phenomena: probabilistic and statistical modeling. Springer Science \& Business Media; 2007.

\bibitem{Castillo2012}
Castillo E. Extreme value theory in engineering. Elsevier; 2012.

\bibitem{Klugman}
Klugman SA, Panjer HH, Willmot GE. Loss models: from data to decisions,  vol. 715. John Wiley \& Sons; 2012.

\bibitem{LeCourtois}
Le Courtois O, Walter C. Extreme financial risks and asset allocation. World Scientific Books; 2014.

\bibitem{Castillo1989}
Castillo E, Galambos J, Sarabia JM. The selection of the domain of attraction of an extreme value distribution from a set of data. In: H\"usler J, Reiss RD, editors. Extreme Value Theory. Lecture
Notes in statistics 1989; 51: 181-190. Springer.

\bibitem{Bassi}
Bassi F, Embrechts P, Kafetzaki M. Risk management and quantile estimation. In: Adler R, Feldman F, Taqqu M, editors. A practical guide to heavy tails 1998; 111-130. Birkh\"auser, Boston.

\bibitem{Asimit}
Asimit AV, Jones BL. Asymptotic tail probabilities for large claims reinsurance of a portfolio of dependent risks. Astin Bulletin 2008; 38(01): 147-159.

\bibitem{Cirillo}
Cirillo P. Are your data really Pareto distributed?. Physica A: Statistical Mechanics and its Applications 2013; 392(23): 5947-5962.

\bibitem{Jayakrishnan}
Nair J, Wierman A, Zwart B. The fundamentals of heavy-tails: properties, emergence, and identification. In: Proceedings of the Acm Sigmetrics International Conference on Measurement and Modeling of Computer Systems 2013; 41(1): 387-388.

\bibitem{Gorge}
Gorge G. Insurance risk management and reinsurance. Lulu. com; 2013.

\bibitem{Resnick2013}
Resnick SI. Extreme values, regular variation and point processes. Springer; 2013.

\bibitem{vonMises}
Von Mises, R. La distribution de la plus grande de n valeurs. Rev. math. Union interbalcanique 1936; 1(1).

\bibitem{Tyszer2012}
Tyszer, J. Object-oriented computer simulation of discrete-event systems (Vol. 10). Springer Science \& Business Media, 2012.

\bibitem{Glasserman2003}
Glasserman, P. Monte Carlo methods in financial engineering (Vol. 53). Springer Science \& Business Media, 2003

\bibitem{Fisher1922}
Fisher RA. On the mathematical foundations of theoretical statistics. Philos Trans Roy Soc Ser A 1922; 222: 309-368.

\bibitem{Rproject}
R Development Core Team. R: A language and environment for statistical computing, R Foundation for Statistical Computing, Vienna, Austria, 2011.
http://www.R-project.org/; [accessed 19.02.16].

\bibitem{Nash}
Nash JC, Varadhan R. Unifying optimization algorithms to aid software system users: optimx for R. Journal of Statistical Software 2011; 43(9): 1-14.

\bibitem{Byrd}
Byrd RH, Lu P, Nocedal J, Zhu CY (1995). A limited memory algorithm for bound constrained optimization. SIAM Journal on Scientific Computing 1995; 16(5): 1190-1208.

\bibitem{Barros}
Barros M, Paula GA, Leiva V. An R implementation for generalized Birnbaum-Saunders distributions. Computational Statistics and Data Analysis 2009; 53: 1511-1528.

\bibitem{Lange}
Lange K. Numerical Analysis for Statisticians. Springer, New York, 2000.

\bibitem{Akaike}
Akaike H. A new look at the statistical model identification. Automatic Control, IEEE Transactions on 1974; 19: 716-723.

\bibitem{Schwarz}
Schwarz G. Estimating the dimension of a model. The Annals of Statistics 1978; 5: 461-464.

\bibitem{Efron}
Efron B. Bootstrap methods: another look at the jackknife. Annals of Statistics 1979; 7(1): 1-26.

\bibitem{Wang}
Wang C, Zeng B, Shao J. Application of bootstrap method in Kolmogorov-Smirnov test. Quality, Reliability, Risk, Maintenance, and Safety Engineering (ICQR2MSE) 2011; 287-91.

\bibitem{Babu}
Babu GJ, Rao CR. Goodness-of-fit tests when parameters are estimated. Sankhya 2004; 66: 63–74.

\bibitem{Kolmogorov}
Kolmogorov AN. Sulla Determinazione Empirica di una Legge di Distribuzione, Giornale dell'Istituto degli Attuari 1933; 4: 83-91.

\bibitem{Smirnov}
Smirnov N. On the estimation of the discrepancy between empirical curves of distribution for two independent samples. Bull. Math. Univ. Moscou, 2, fasc 2, 1939.

\bibitem{Castillo2005}
Castillo E, Hadi AS, Balakrishnan N, Sarabia JM. Extreme value and related models with applications in engineering and science. John Wiley \& Sons; 2005.

\bibitem{DGCL}
Spanish Ministry of Finance and Public Admin. Local Government in  Spain; 2008.
http://www.seap.minhap.gob.es/web/publicaciones.html; [accessed 19.02.16].

\bibitem{Benito2004}
Benito B., Bastida F. The determinants of the municipal debt policy in Spain. Journal of Public Budgeting, Accounting and Financial Management 2004; 16(4): 525-558.

\bibitem{INE}
Spanish National Statistics Institute (INE), Spanish Official Municipal Register. http://www.ine.es; [accessed 19.02.16].

\bibitem{Montesinos2000}
Montesinos V, Vela JM. Governmental accounting in Spain and the European Monetary Union: A critical perspective. Financial Accountability \& Management 2000; 16(2): 129-150.

\bibitem{Alba2003}
Alba C, Navarro C. Twenty-five years of democratic local government in Spain. In: Reforming local government in Europe 2003; 197-220. VS Verlag f\"ur Sozialwissenschaften.

\bibitem{Sole2006}
Solé-Ollé A. The effects of party competition on budget outcomes: Empirical evidence from local governments in Spain. Public Choice 2006; 126(1-2): 145-176.

\bibitem{Cabases2007}
Cabasés F, Pascual, P, Vallés J. The effectiveness of institutional borrowing restrictions: Empirical evidence from Spanish municipalities. Public Choice 2007; 131(3-4): 293-313.

\bibitem{Bastida2009}
Bastida F, Benito B, Guillamón MD. An empirical assessment of the municipal financial situation in Spain. International Public Management Journal 2009; 12(4): 484-499.

\bibitem{Hita2011}
Hita FC, Orayen RE, Arzoz PP. Municipal indebtedness in Spain revisited: the impact of borrowing limits and urban development. In: XVIII Encuentro de economía pública (p. 9), 2011.

\bibitem{GarciaSanchez2011}
Garcia-Sanchez IM, Prado-Lorenzo JM, Cuadro-Ballesteros B. Do progressive governments undertake different debt burdens? Partisan vs. electoral cycles. Revista de Contabilidad, 2011; 14 (1): 29-57.

\bibitem{GarciaSanchez2012}
Garcia-Sanchez IM, Mordan N, Prado-Lorenzo J. Effect of the political system on local financial condition: Empirical evidence for Spain's largest municipalities. Public Budgeting \& Finance 2012; 32(2): 40-68.

\bibitem{Lopez2012}
Lopez-Hernandez AM, Zafra-Gomez  JL, Ortiz-Rodriguez D. Effects of the crisis in Spanish municipalities' financial condition: an empirical evidence (2005-2008). International Journal of Critical Accounting 2012; 4(5-6): 631-645.

\bibitem{Almendral2013}
Almendral VR. The Spanish legal framework for curbing the public debt and the deficit. European Constitutional Law Review 2013; 9(02): 189-204.

\bibitem{Benito2015}
Benito B, Vicente C, Bastida F. The Impact of the Housing Bubble on the Growth of Municipal Debt: Evidence from Spain. Local Government Studies 2015; 41(6): 997-1016.

\bibitem{MFPA}
Spanish Ministry of the Finance and Public Administrations. http://www.minhap.gob.es/; [accessed 19.02.16].

\bibitem{Clauset2009}
Clauset A, Shalizi CR, Newman MEJ. Power-law distributions in empirical data. SIAM Rev 2009; 51(4): 661-703.

\bibitem{Hill1975}
Hill BM. A simple general approach to inference about the tail of a distribution. The Annals of Statistics 1975; 3(5): 1163-1174.

\bibitem{Brakman}
Brakman S, Garretsen H, van Marrewikj  C, van de Berg M. The return of Zipf: towards a further understanding of the Rank-Size distribution.
Journal of Regional Science 1999; 39(1): 182-213.

\bibitem{Gabaix1999a}
Gabaix X. Zipf's law and the growth of cities. American Economic Review 1999; 89(2): 129-132.

\bibitem{Gabaix1999b}
Gabaix X. Zipf's law for cities: An explanation. The Quarterly Journal of Economics 1999; 114: 739-767.

\bibitem{Urzua}
Urz\'ua CM. A simple and efficient test for Zipf's law. Economics Letters 2000, 66: 257-260.

\bibitem{Ioannides}
Ioannides YM, Overman HG.  Zipf's law for cities: an empirical examination. Regional Science and Urban Economics 2003; 33: 127-137.

\bibitem{Fujiwara}
Fujiwara Y, Guilmi CD, Aoyama H, Gallegati M, Souma W. Do Pareto-Zipf and Gibrat laws hold true? An analysis with European firms. Physica A: Statistical Mechanics and its Applications  2004; 335: 197-216.

\bibitem{Gabaix2004}
Gabaix X, Ioannides YM. The evolution of city size distributions. In: Henderson JV, Thisse JF, editors. Handbook of Regional and Urban Economics 2004; 4: 2341-2378. Elsevier, Amsterdam.

\bibitem{Anderson}
Anderson H, Ge Y. The size distribution of Chinese cities. Regional Science and Urban Economics 2005; 35: 756-776.

\bibitem{Cordoba}
C\'ordoba JC. On the distribution of city sizes. Journal of  Urban Economics 2008; 63: 177-197

\bibitem{Lomax1954}
Lomax KS. Business failures; another example of the analysis of failure data. J Am Stat Assoc 1954; 49: 847-852.

\bibitem{Johnson1994}
Johnson NL, Kotz S, Balakrishnan N. Continuous univariate distributions, vol.1. New York: John Wiley; 1994.

\bibitem{Fisk1961}
Fisk PR. The graduation of income distributions. Econometrica 1961; 29: 171-185.

\bibitem{Burr1942}
Burr IW. Cumulative frequency functions. The Annals of Mathematical Statistics  1942; 13(2): 215-232.

\bibitem{Singh1976}
Singh S, Maddala G. A function for size distribution of incomes. Econometrica 1976; 44 (5): 963-970.

\bibitem{Dagum1975}
Dagum C. A model of income distribution and the conditions of existence of moments of finite order. In: Proceedings of the 40th session of the International Statistical Institute 1975; 46: 199-205.

\end{thebibliography}

\end{document}